\documentclass[10pt]{article}      

\usepackage[british]{babel}

\usepackage{amsmath,amssymb,amsfonts,amsthm,amscd,mathtools} 
\usepackage{graphics}                 
\usepackage{color}                    
\usepackage{mdframed}                
\usepackage{dsfont, bbold, mathrsfs}  
\usepackage{nccmath}                  

\usepackage{url}         

\usepackage{enumerate}  

\makeatletter
  \newcommand{\eqnum}{\leavevmode\hfill\refstepcounter{equation}\textup{\tagform@{\theequation}}} 
\makeatother

\usepackage[hidelinks]{hyperref}

\usepackage{tikz}
\usetikzlibrary{arrows,decorations.markings,patterns}

\theoremstyle{plain}
\newtheorem{theorem}{Theorem}[section]
\newtheorem{lemma}[theorem]{Lemma}
\newtheorem{prop}[theorem]{Proposition}
\theoremstyle{definition}

\theoremstyle{remark}
\usepackage{etoolbox} 
\AtEndEnvironment{rem}{\qed}

\numberwithin{equation}{section}
\newtheorem{remark}[theorem]{Remark}
\definecolor{aliceblue}{rgb}{0.94, 0.97, 1.0} 
\definecolor{azuremist}{rgb}{0.94, 1.0, 1.0} 
\definecolor{antiquewhite}{rgb}{0.98, 0.92, 0.84} 
\definecolor{ivory}{rgb}{1.0, 1.0, 0.94}
\newtheorem{mdexample}[theorem]{Example.}
\newenvironment{example}%
  {\begin{mdframed}[backgroundcolor=ivory]\begin{mdexample}}%
  {\end{mdexample}\end{mdframed}}

\newcommand{\bw}{\mathrm{bw}}
\newcommand{\fw}{\mathrm{fw}}

\newcommand{\grad}{\nabla}
\newcommand{\iso}{\mathrm{iso}}

\newcommand{\sym}{\mathrm{sym}}
\newcommand{\asym}{\mathrm{asym}}

\DeclareMathOperator{\Prob}{\mathrm{Prob}}

\newcommand{\super}[1]{^{\scriptscriptstyle{(#1)}}}
\newcommand{\supernul}[1]{^{{\scriptscriptstyle(#1)}0}}

\def\tp{^\mathsf{T}}

\def\se{\mathsf{e}}
\def\sa{\mathsf{a}}

\def\J{\mathcal{J}}
\def\I{\mathcal{I}}

\def\H{\mathcal{H}}
\def\L{\mathcal{L}}

\def\NN{\mathbb{N}}
\def\bigoh{\mathcal{O}}

\def\R{\mathcal{R}}
\def\S{\mathcal{S}}
\def\fS{\mathfrak{S}}

\def\V{\mathcal{V}}

\def\X{\mathcal{X}}

\def\Q{\mathcal{Q}}
\def\RR{\mathbb{R}}

\def\NN{\mathbb{N}}

\date \today
\title{Anisothermal Chemical Reactions: Onsager-Machlup and Macroscopic Fluctuation Theory}
\author{D.R.\ Michiel Renger\thanks{WIAS, Mohrenstrasse 39, 10117 Berlin, Germany. Email: \href{mailto:renger@wias-berlin.de}{renger@wias-berlin.de}}}

\begin{document}
\maketitle

\begin{abstract}
We study a micro and macroscopic model for chemical reactions with feedback between reactions and temperature of the solute. The first result concerns the quasipotential as the large-deviation rate of the microscopic invariant measure. The second result is an application of modern Onsager-Machlup theory to the pathwise large deviations, in case the system is in detailed balance. The third result is an application of macroscopic fluctuation theory to the reaction flux large deviations, in case the system is in complex balance.
\end{abstract}

\section{Introduction}
\label{sec:intro}

\subsection{Historic overview}
\label{subsec:history}
Recent decades have seen a great progress in non-equilibrium thermodynamics. In a sense, the discovery of Wasserstein gradient flows~\cite{JKO1998} for Fokker-Planck equations meant a strengthening of the second law: not only does the free energy decrease along the dynamics, but the decrease of free energy fully determines the dynamics. Many other evolution equation have since been shown to be gradient flows driven by free energy functionals, see for example~\cite{AGS2008}. We mention in particular the works~\cite{Maas2011,Mielke2012a,Chow2012} for gradient flows on discrete spaces and gradient flows for the chemical reaction-rate equation~\cite{LieroMielke2013,MaasMielke2020}. 

As often in thermodynamics, many of these structures are actually related to the statistics of microscopic particle systems, which adds to their physical validity. In this respect, the Wasserstein metric tensor and its corresponding optimal transport formulation were already derived from large deviations of Brownian particles in~\cite{Dawson1987} and \cite{Leonard2007} respectively.
More recent works show that gradient flow structures can be uniquely derived from pathwise large deviations via what may be called a modern version of Onsager-Machlup theory. Whereas Onsager's original paper~\cite{Onsager1931I} already deals with gradient flow structures for the chemical reaction-rate equation, his later work with Machlup~\cite{Onsager1953I} makes the connection to fluctuations of a microscopic particle system via the pathwise large deviations. More modern and precise relations of this type were made for exclusion processes and Brownian particles in~\cite{BertiniEtAl2004,ADPZ2011,ADPZ2012,Duong2013a,EMR2015}. Pursuing a similar program for particles on discrete spaces required generalised gradient flows, allowing for nonlinear response theory. Although such structures were already explored in the classic~\cite{Marcelin1915}, its relation to large deviations was discovered in~\cite{MielkePeletierRenger2014}, and worked out for chemical reactions in~\cite{MielkePattersonPeletierRenger2016}. Notably, it turns out that the above-mentioned (linear) gradient structures for the chemical reaction-rate equation correspond to white noise fluctuations, whereas the nonlinear structure from \cite{MielkePattersonPeletierRenger2016} corresponds to more physical fluctuations due to reactions on the molecular scale.

As prescribed by Onsager's reciprocity relations, this connection between gradient flows and large deviations can only be carried out for microscopic systems in detailed balance, which loosely corresponds to thermodynamically closed systems on the macro scale. Hence, for systems that are driven out of equilibrium by an external force, one looks for thermodynamically consistent couplings of gradient flows with Hamiltonian systems or other `rotational' motions. In connection with large deviation there are two main directions in the literature. First, the work~\cite{KLMP2018math} studies how large deviations are related to GENERIC~\cite{Grmela1997}. This allows for microscopic particle systems that are in detailed balance with an additional drift that is approximately deterministic. In the context of chemical reactions, such drifts can be obtained by chemical reactions of lower-order concentrations on a faster time scale~\cite{Renger2018}. The second approach is Macroscopic Fluctuation Theory (MFT), based on the orthogonal decomposition of thermodynamics forces~\cite{Bertini2015MFT}. As for the setting of gradient flows described above, the connection with large deviations required a generalisation to nonlinear response relations, which was carried out in~\cite{PattersonRengerSharma2021TR}. The same paper shows how this program can be applied to reacting particle systems in case of complex balance (which is more general than detailed balance).

So far, most literature on this topic is restricted to constant temperature. In this paper we describe and explain how the above program can be applied to temperature-dependent chemical reactions. We consider a finite isolated box with a well-mixed solute, so that we may ignore spatiality. The solute contains reactants undergoing a set of reactions, and as temperature-dependent reactions take place, heat is being exchanged with the solute. Hence there is a coupling between the evolution of concentrations and temperature. This phenomenon can be modelled on two different levels.

\subsection{Macroscopic model}
\label{sec:macro}

The solute contains reactants of species $x\in\X$ undergoing reactions $r\in\R$. A reaction $r=(\alpha\to\alpha')$ takes $\alpha\in\NN_0^\X$ molecules and produces $\alpha'\in\NN_0^\X$ molecules. We shall assume that the reaction network is `reversible', meaning that $(\alpha\to\alpha')\in\R\implies (\alpha'\to\alpha)\in\R$, see for example~\cite[Def.~2.2]{AndersonCraciunKurtz2010}. This allows to decompose the set of reactions into forward and backward reactions $\R=\R_\fw\cup\R_\bw$. For a reaction $r=(\alpha\to\alpha')$ we define its backward reaction $\bw(r)\coloneqq (\alpha'\to\alpha)$, and we write $\alpha\super{r}\coloneqq \alpha$ and $\alpha^{\bw(r)}=\alpha'$.

Further, we assume that the reaction rates $k_r$ for reaction $r$, depending on concentrations $\rho\in\lbrack0,\infty)^\X$ of species and temperature $\theta\in\lbrack0,\infty)$ of the solute, obey the law of mass action:
\begin{align}
  k_r(\rho,\theta)\coloneqq \kappa_r A_r(\theta) B_r(\rho),
  &&\text{with}&&
  B_r(\rho)\coloneqq \rho^{\alpha\super{r}}\coloneqq \prod_{x\in\X}\rho_x^{\alpha\super{r}_x},
\label{eq:reaction rates}
\end{align}
and $\kappa_r$ are given constants.

The temperature-dependent factors $A_r(\theta)$ obey the (modified) Arrhenius or Eyring law as follows. Each species $x\in\X$ corresponds to an energy level $\se_x\geq0$, so that a complex $\alpha$ stores a total amount $\se\cdot\alpha=\sum_{x\in\X}\se_x\alpha_x$ of chemical energy. A reaction $r$ needs to go through a transitional state with chemical energy $\sa_r$, and so the activation energy or energy barrier of the reaction is $\sa_r - \se\cdot\alpha\super{r}$, see Figure~\ref{fig:energy levels} for a schematic representation with a one-dimensional reaction coordinate. From this perspective it is logical to assume that $\sa_{\bw(r)}=\sa_r$. According to the Arrhenius or Eyring law,
\begin{equation}
  A_r(\theta)\coloneqq  \theta^q\exp\big( - \frac{\sa_r - \se\cdot\alpha\super{r}}{ k_B \theta}\big),
\label{eq:Arrhenius}
\end{equation}
for some $q\in(-1,1\rbrack  $, where $k_B$ is the Boltzmann constant.

\begin{figure}[h!]
\centering
\begin{tikzpicture}
\tikzstyle{every node}=[font=\scriptsize];
  \draw[<->](0,4) node[anchor=east]{energy level}--(0,0)--(4,0) node[anchor=west]{reaction coordinate};
  \draw[dotted](0,2) node[anchor=east]{$\se\cdot\alpha\super{r}$}--(4,2);
  \draw[dotted](1,2)--(1,0) node[anchor=north]{$\alpha\super{r}$};
  \draw[dotted](0,1) node[anchor=east]{$\se\cdot\alpha^{\bw(r)}$}--(4,1);
  \draw[dotted](3,1) -- (3,0) node[anchor=north]{$\alpha^{\bw(r)}$};
  \draw[dotted](0,3) node[anchor=east]{$\sa_r$}--(4,3);
  \draw(0.5,2.5).. controls (0.7,2) and (0.85,2) .. (1,2) .. controls(1.5,2) and (1.6,3) .. (2,3) .. controls (2.4,3) and (2.5,1) .. (3,1) .. controls (3.15,1) and (3.3,1) .. (3.5,1.5);
  \draw[->](1,2)-- node[midway,anchor=west]{activation energy} (1,3);
  \draw[->](3,2)-- node[midway,anchor=west]{energy release} (3,1);
\end{tikzpicture}
\caption{Energy levels corresponding to a reaction $r$.}
\label{fig:energy levels}
\end{figure}
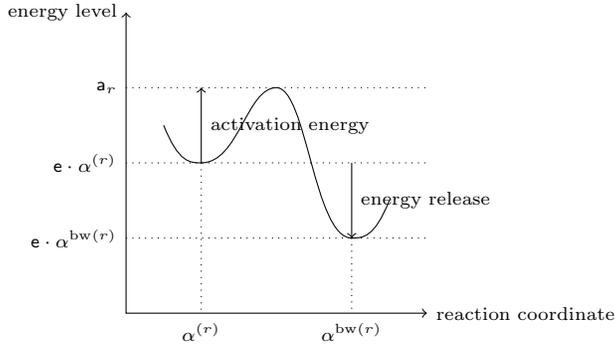

The effect of a reaction $r$ occurring with some rate $j_r\geq0$ is twofold. First, reactants $\alpha\super{r}$ are annihilated and the products $\alpha^{\bw(r)}$ are created, causing a change in concentrations with rate $(\alpha^{\bw(r)}-\alpha\super{r})j_r$. To simplify notation one encodes this principle in the stoichiometric matrix:
\begin{align}
  \Gamma\coloneqq \lbrack \gamma\super{r} \rbrack_{r\in\R_\fw} \in \RR^{\X\times\R_\fw}, 
  &&
  \gamma\super{r}\coloneqq \alpha^{\bw(r)} - \alpha \super{r}
\label{eq:Gamma}
\end{align}
The second effect of a reaction $r$ is a change in energy levels such that heat $-\se\cdot\gamma\super{r}$ is transferred to (or extracted from) the solute. Assuming a homogeneous and well-mixed solute with specific heat $c_H$, this causes the solute temperature to change with rate $-c_H^{-1}\se\cdot\gamma\super{r}j_r$.

Concluding, the evolution of the concentrations and temperature can now be written as:
\begin{align}
  \dot\rho(t) &=  \Gamma J^0 \big(\rho(t),\theta(t)\big), \label{eq:ODE}\\
  \dot\theta(t) &= -\frac1{c_H} \se\cdot\Gamma J^0 \big(\rho(t),\theta(t)\big) \notag
\end{align}
with given initial conditions $\rho(0)=\rho^0,\theta(0)=\theta^0$. Here we used the abbreviation
\begin{equation*}
  J^0_r(\rho,\theta)\coloneqq k_r(\rho,\theta)-k_{\bw(r)}(\rho,\theta)=\kappa_r A_r(\theta) B_r(\rho) - \kappa_{\bw(r)} A_{\bw(r)}(\theta) B_{\bw(r)}(\rho),
\end{equation*}
which is more than just notationally convenient: $J_r(\rho,\theta)$ describes the instantaneous net reaction flux through reaction channel $r$ (sometimes also called the traffic in the literature). These fluxes are an essential ingredient to understand non-dissipative effects in MFT, as we shall see in Section~\ref{sec:MFT}.

Note that \eqref{eq:ODE} clearly conserves the total energy
\begin{equation}
  E^0\coloneqq \se\cdot\rho^0 + c_H \theta^0\equiv\se\cdot\rho(t) + c_H \theta(t),
\label{eq:total energy}
\end{equation}
so with fixed $\rho^0$ and $\theta^0$, the temperature is really a function of concentration. To shorten notation we can therefore omit dependencies on $\theta$, e.g. by a slight abuse $k_r(\rho)\coloneqq k_r\big(\rho,\mfrac{E^0-\se\cdot\rho}{c_H}\big)$.

\subsection{Microscopic model}

To reflect molecular fluctuations, the same phenomenon as above can be described by the usual microscopic model~\cite{Kurtz1972,AndersonCraciunKurtz2010}, with a small adaptation to account for temperature effects. 
The order of the total particle number in the system is controlled by the parameter $V$. This may physically be interpreted as the volume \emph{ratio} of the system size with respect to particle size (although these sizes will play no further role). Concentrations are defined as particle numbers, normalised by $V$. Initially we fix the temperature $\Theta\super{V}(0)\equiv\theta^0\in\lbrack0,\infty)$ and the concentrations $\rho\super{V}(0)=\rho\supernul{V}\in(\NN_0/V)^\X$ such that the initial concentrations converge to some limit:
\begin{equation*}
  \rho\supernul{V}\to \rho^0\qquad\text{as } V\to\infty.
\end{equation*}
A reaction $r$ occurs randomly with jump rate
\begin{align}
  Vk_r\super{V}(\rho,\theta)&\coloneqq V\kappa_r A\super{V}_r(\theta) B\super{V}_r(\rho),
  \label{eq:micro jump rates}\\
  A\super{V}_r(\theta) &\coloneqq  A_r(\theta)\mathds1_{\{Vc_H \theta\geq\se\cdot\gamma\super{r}\}},\notag\\
  B\super{V}_r(\rho)   &\coloneqq \frac{1}{V^{\lvert\alpha\super{r}\rvert_1}}\frac{(V\rho)!}{(V\rho-\alpha\super{r})!} \mathds1_{\{V\rho\geq\alpha\super{r}\}},\notag
\end{align}
denoting $\alpha!\coloneqq \prod_{x\in\X}\alpha_x!$. The expression for $B\super{V}_r$ is standard, based on combinatinorial arguments, and known to approximate the mass-action factor $B_r$~\cite{AndersonCraciunKurtz2010}. The two indicators -- where the inequalities are meant coordinatewise -- ensure that reactions leading to negative temperature respectively negative concentrations do not take place. However, zero temperatures or zero concentrations are not forbidden; this will

The effect of a microscopic reaction $r$ is a small change $V^{-1}\gamma\super{r}$ in concentrations and the release (or extraction) of a heat package of size $-V^{-1}\se\cdot\gamma\super{r}$. The pair $(\rho\super{V}(t),\Theta\super{V}(t))$ is then a Markov jump process in $\RR^\X \times\RR$ with generator:
\begin{align}
  (\hat\Q\super{V}F)(\rho,\theta) = V\sum_{r\in\R}\kappa_r A_r\super{V}(\theta) B\super{V}_r(\rho) \Big\lbrack F\big(\rho+\tfrac1V\gamma\super{r}, \theta - \tfrac1{Vc_H}\se\cdot\gamma\super{r}\big)-F(\rho,\theta)\Big\rbrack.
\label{eq:hatQ}
\end{align}
Just as for the macroscopic scale, the total energy is almost surely conserved:
\begin{equation}
  E\supernul{V}\coloneqq \se\cdot\rho\supernul{V} + c_H \theta^0 \equiv \se\cdot\rho\super{V}(t) + c_H \Theta\super{V}(t),
\label{eq:total energy micro}
\end{equation}
and so we may omit temperature dependencies.

By the classical Kurtz limit, the random process $(\rho\super{V},\theta\super{V})$ converges (weakly in Skorohod space) to the macroscopic trajectory $(\rho,\theta)$ solving~\eqref{eq:ODE}, \eqref{eq:total energy}. In this sense, the microscopic model shows vanishing random fluctuations around the macroscopic evolution.

Naturally it is desirable to include spatiality in our model. However, there are not many physically reasonable models in the literature that describe microscopic heat transport in liquids or gases. The classic KMP model \cite{KMP1982} and BEP models \cite{PeletierRedigVafayi2014} are specifically for crystals. The recent kinetic exclusion process~\cite{GutierrezHurtado2019} seems more suitable for solutes, but a detailed knowledge of the fluctuations is still unknown.

\subsection{Results and overview}
\label{subsec:goal}

This paper contains three main contributions.

I. In Section~\ref{sec:QP}, we study the invariant measure $\Pi\super{V}$ for the microscopic system, which can be constructed explicitly for a very simple example $\mathsf{A}\rightleftharpoons\mathsf{B}$. For this example as well as in a more general setting we show that $\Pi\super{V}$ satisfies a large-deviation principle as $V\to\infty$, formally written as (omitting temperature-dependencies):
\begin{equation}
  \Pi\super{V}\big( \rho\super{V} \approx \rho \big) \sim \exp\big(-V \V(\rho)\big)
\label{eq:LDP Pi}
\end{equation}
where $\rho$ is now a dummy variable, the \emph{quasipotential} is given by
\begin{align}
  \V(\rho) \coloneqq  \,&\S(\rho\mid\pi) - \frac{c_H}{k_B}\log\theta + \text{const.}, \qquad
  \theta\coloneqq \frac{E^0-\se\cdot\rho}{c_H}, \label{eq:QP}\\
  &\S(\rho\mid\pi)\coloneqq \sum_{x\in\X} s(\rho_x\mid \pi_x), \notag\\
  &s(a\mid b)\coloneqq  \begin{cases}
      a\log \frac{a}{b}-a+b, & a,b>0,\\
      b,                                             & a = 0,\\
      \infty,                                        & b=0, a>0 \text{ or } a<\infty,
    \end{cases}
    \label{eq:Boltzmann function}
\end{align}
and $\pi\in\RR^\X$ is the steady state concentration of the isothermal ODE, i.e.~\eqref{eq:ODE} with $A_r\equiv1$. This result only holds under the assumption of \emph{isothermal detailed balance}, or under the assumption of \emph{isothermal complex balance} with negligible or constant transition energy levels $\sa_r$. Physically, $\S(\rho\mid\pi)$ encodes the entropic contribution and $-\frac{c_H}{k_B}\log\theta$ the thermal one.

II. In Section~\ref{sec:QP}, we study a pathwise large-deviations principle and in Section~\ref{sec:Onsager-Machlup} we show that, under the assumption of detailed balance, it satisfies an Onsager-Machlup principle, following~\cite{MielkePeletierRenger2014}:
\begin{align}
  &\Prob\big( \rho\super{V}\approx \rho \big)\sim
  \exp\bigg(-V\times \notag\\
  &\quad\bigg\lbrack
     {\textstyle
       \int_0^T\!\hat\Psi\big(\rho(t),\dot\rho(t)\big)\,dt
     }
     {\textstyle
       +
       \int_0^T\!\hat\Psi^*\big(\rho(t),-\tfrac12\grad\V(\rho(t))\big)\,dt
       +
       \tfrac12\V\big(\rho(T)\big)-\tfrac12\V\big(\rho(0)\big)
     }\,
  \bigg\rbrack\bigg),
\label{eq:Onsager-Machlup}
\end{align}
where $\V$ is the quasipotential from~\eqref{eq:QP}, $\hat\Psi$ is a non-negative \emph{dissipation potential}, and $\hat\Psi^*$ is its convex dual with respect to the second variable. These microscopic fluctuations contain additional information about the macroscopic dynamics, for which the whole exponent of \eqref{eq:Onsager-Machlup} is $0$. Hence along solutions of~\eqref{eq:ODE} (always taking the gradient of $\hat\Psi^*$ with respect to the second entry):
\begin{equation}
  \dot\rho(t) = \grad\hat\Psi^*\big( \rho(t),-\tfrac12\grad\V(\rho(t))\big).
\label{eq:GF}
\end{equation}
The pair $\hat\Psi,\hat\Psi^*$ generalises dual squared norms $\tfrac12\lVert\cdot\rVert^2, \tfrac12\lVert\cdot\rVert^2_*$. We refer to Section~\ref{sec:Onsager-Machlup} for the details, but already mention that $\hat\Psi$ will indeed not be quadratic. Hence \eqref{eq:GF} is a \emph{nonlinear} relation between the force $-\tfrac12\grad\V$ and velocities $\dot\rho,\dot\theta$, which is also interpreted as a nonlinear generalisation of a gradient flow, driven by the quasipotential $\V$, and this gradient flow is uniquely defined by the Onsager-Machlup decomposition~\eqref{eq:Onsager-Machlup}. For the expected path~\eqref{eq:ODE} the exponent in \eqref{eq:Onsager-Machlup} sums up to $0$, yielding a balance between free energy loss and dissipation.

III. We then look beyond detailed balance, but in order to identify the quasipotential $\V$ we need to assume isothermal complex balance and constant $\sa_r$. Since a breaking of detailed balance can and will result in the occurrence of divergence-free fluxes, a force or energy balance can only be found when taking reaction fluxes into account. In this context, a reaction flux $J\super{V}_r(t)$ counts the net number of reactions through channel $r$ taking place at time $t$. Looking at large deviation of fluxes we enter the field of \emph{Macroscopic Fluctuation Theory} (MFT)~\cite{Bertini2015MFT}. We first derive the pathwise large deviations of the fluxes in Section~\ref{sec:LDP}. Following~\cite{RengerZimmer2021,PattersonRengerSharma2021TR}, one distinguishes between the \emph{symmetric force} $F^\sym=-\frac12\Gamma\tp\grad\V$ corresponding to the gradient flow part of the dynamics and the \emph{antisymmetric force} $F^\asym$. Note that $F^\sym$ is indeed similar to the force $-\tfrac12\grad\V$ from~\eqref{eq:GF}, but now with $\Gamma\tp$ which appears because the force acts on fluxes rather than velocities. In Section~\ref{sec:MFT} we first decompose the flux large deviations as: 
\begin{multline}
  \Prob\big( (\rho\super{V},J\super{V})\approx (\rho,j) \big) \\
  \sim \exp\bigg(-V \int_0^T\!\Big\lbrack
    \Psi\big(\rho(t),j(t)\big) 
  + \Psi^*\big(\rho(t),F^\sym(\rho(t))+F^\asym(\rho(t))\big) \\
    - F^\sym\big(\rho(t)\big)\cdot j(t) -F^\asym\big(\rho(t)\big)\cdot j(t) \,\Big\rbrack\,dt\bigg),
\label{eq:L decomp1}
\end{multline}
where $\Psi,\Psi^*$ are again non-quadratic dissipation potentials. Note that $-\int_0^T\!F^\sym\cdot j\,dt = \tfrac12\int_0^T\!\grad\V\cdot\Gamma j\,dt = \tfrac12\int_0^T\!\grad\V\cdot \dot\rho\,dt=\tfrac12\V(\rho(T))-\tfrac12\V(\rho(0))$ as in \eqref{eq:Onsager-Machlup}. By contrast, the term $-\int_0^T\!F^\asym\cdot j\,dt$ representing the work done by the antisymmetric force is path-dependent.

Next we use the notion of \emph{generalised orthogonality} developed in \cite{KaiserJackZimmer2018,RengerZimmer2021,PattersonRengerSharma2021TR} to further decompose 
\begin{align*}
  \Psi^*\big(\rho,F^\sym(\rho)+F^\asym(\rho)\Big) &= \Psi^*\big(\rho,F^\sym(\rho)\big) + \Lambda_\sym^\asym(\rho) \notag\\
  &= \Psi^*\big(\rho,F^\asym(\rho)\big) + \Lambda_\asym^\sym(\rho). 
\end{align*}
The two non-negative terms $\Lambda_\sym^\asym$ and $\Lambda_\asym^\sym$ are interpreted as generalisations of Fisher informations. Together with \eqref{eq:L decomp1} the flux large deviations split into terms for the symmetric (gradient flow) dynamics and terms for the antisymmetric dynamics, and since $\Psi^*$ is not quadratic, there are two different ways to do so. 

These decompositions are helpful to obtain estimates on the two work terms $\int_0^T\!F^\sym\cdot j\,dt$ and $\int_0^T\!F^\asym\cdot j\,dt$ and to extract compactness of paths and fluxes. The exact expressions, interpretation and applications will be given in Section~\ref{sec:MFT}.

\subsection{Three mathematical subtleties}
\label{subsec:subtleties}

Although we mostly apply existing mathematical techniques to a new setting, we encounter a number of mathematical subtleties that can not be disregarded as mere technicalities.
\begin{enumerate}

\item First of all, we are only able to construct the invariant measure $\Pi\super{V}$ for the very simple system of reactions $\mathsf{A}\rightleftharpoons\mathsf{B}$. We believe that a similar formula should hold for more general reaction networks, but the exact expression is not clear at this stage. For the more general setting we derive the quasipotential $\V$; this corresponds to the large-deviation principle~\eqref{eq:LDP Pi} of the invariant measure without actually knowing the invariant measure. Since the quasipotential~\eqref{eq:QP} decomposes into chemical and thermal factors, this suggests that (at least approximately) the invariant measure $\Pi\super{V}$ factorises into chemical and thermal factors, which is indeed confirmed for the simple example.

\item In usual reaction network theory one can impose macroscopic detailed or complex balance of a steady state in order to derive the invariant measure~\cite{AndersonCraciunKurtz2010}. This would be problematic when coupled to temperature, since the steady state does not have an explicit form; see the calculation at the end of Subsection~\ref{subsec:QP general}. Instead we assume macroscopic detailed and complex balance for the reaction network without a temperature coupling. In addition to (isothermal) complex balance we need to assume that $\sa_r$ is constant in $r$, for example corresponding to ``barrierless'' reactions. It should be said that the proof of our Proposition~\ref{prop:QP} shows that for isothermal complex balanced reactions with non-constant $\sa_r$, the quasipotential is generally not of the form~\eqref{eq:QP}.

\item Although rigorous proofs of pathwise large deviations tend to be quite technical, the rate functional can usually be derived formally by a simple calculation. The setting of our paper is an example where such naive calculation may fail. The reason is that once a path $\rho(t)$ reaches the zero temperature state, it can no longer escape that state with finite large-deviation cost, see~\cite{AAPR2021TR} and Subsection~\ref{subsec:cold death}. This phenomenon can be interpreted as a different type of cold death (different from the notion of cold death of the universe when it reaches maximal entropy). To circumvent this phenomenon we shall assume sufficient initial total energy, so that by energy conservation the zero-temperature state can never be reached.

\end{enumerate}

\section{Large deviations and the cold death}
\label{sec:LDP}

We first recall the pathwise large deviations for the process $\rho\super{V}(t)$, then comment on the cold death phenomenon, and finally introduce the pathwise large deviations for the fluxes.

\subsection{Pathwise large deviations}
\label{subsec:LDP conc}

The pathwise large deviations for the concentrations quantify the exponential rate of convergence of the trajectories, formally written as\footnote{For the sake of brevity we omit the rigorous definition, see for example~\cite[Th.~1]{AAPR2021TR}.}
\begin{equation}
  \Prob\Big( \rho\super{V} \approx \rho \Big) \sim e^{-V \I_{\lbrack0,T\rbrack}(\rho)} \qquad \text{as } V\to\infty,
\label{eq:LDP}
\end{equation}
where $\rho=(\rho(t))_{t\in\lbrack0,T\rbrack}$ is now an arbitrary trajectory. The functional $\I_{\lbrack0,T\rbrack}$ is called the rate functional; it is always non-negative, and $0$ for the macroscopic trajectory $\rho$ solving~\eqref{eq:ODE}. Physically, it quantifies the total free energy needed to deviate from the expected trajectory. 

The large-deviation principle can be formally calculated via standard procedures, see for example~\cite{Kipnis1999}, \cite{Feng2006} and the application to chemical reactions in~\cite{MielkePattersonPeletierRenger2016}. There is a large body of literature dedicated to making this statement rigorous~\footnote{See for example \cite{PR2019} and \cite{AAPR2021TR} and the cited papers therein.}, where finally a sharp condition was found in~\cite{AAPR2021TR}. As mentioned in Subsection~\ref{subsec:subtleties}, this condition is very relevant to our setting, and we need an additional assumption to prevent cold death. We first state the precise result and dedicate the next subsection to the discussion of the assumption.

Let
\begin{equation}
  \fS\coloneqq \Big\{ \rho=\rho^0 + \Gamma w: w\in \RR^{\R_\fw}  \text{ such that } \rho \in \lbrack0,\infty)^\X \text{ and } \theta=\mfrac{E^0-\se\cdot\rho}{c_H}\geq0  \Big\}
\label{eq:fS}
\end{equation}
be all non-negative concentrations that are attainable from the initial condition $(\rho^0,\theta^0)$ through reactions in $\R$. Similarly, let
\begin{align}
  \theta^- \coloneqq  \inf\Big\{\mfrac{E^0-\se\cdot\rho}{c_H}:\rho\in\fS \Big\},
&&
  \theta^+ \coloneqq  \sup\Big\{\mfrac{E^0-\se\cdot\rho}{c_H}:\rho\in\fS \Big\}
\label{eq:temp bounds}
\end{align}
be the minimal and maximal attainable temperatures. The large-deviation result is the following.

\begin{theorem}[\cite{AAPR2021TR}] 
Assume $\theta^0,\rho^0$ and $\R$ are such that $\fS$ is bounded and $\theta^->0$. Then $\rho\super{V}$ satisfies a large-deviation principle in $D(0,T;\RR^\X)$ with rate functional
\begin{equation*}
  \I_{\lbrack0,T\rbrack}(\rho)\coloneqq 
  \begin{cases}
    \int_0^T\!\hat\L\big(\rho(t),\dot\rho(t)\big)\,dt, & \rho \in W^{1,1}(0,T;\RR^\X), \quad \rho(0)=\rho^0,\\
    \infty, &\text{otherwise},
  \end{cases}
\end{equation*}
where
\begin{equation}
  \hat\L(\rho,u)\coloneqq 
      \inf_{\substack{j\in\RR^\R: \\ u=\sum_{r\in\R}j_r \gamma\super{r}}} \S\big(j\mid k(\rho)\big),
\label{eq:hatL}
\end{equation}
and $\S\big(j\mid k(\rho,\theta)\big) \coloneqq  \sum_{r\in\R} s\big(j_r\mid k_r(\rho)\big)$, recalling~\eqref{eq:Boltzmann function}.
\label{th:LDP conc}
\end{theorem}

Here and throughout the paper, a hat (\string^) will be used to distinguish functions that depend on concentrations $\rho$ (and its time derivative) from functions that also depend on fluxes, introduced in Subsection~\ref{subsec:flux ldp}.

\subsection{Boundary escape and the cold death}
\label{subsec:cold death}

The required boundedness of $\fS$ in Theorem~\ref{th:LDP conc} implies that the rates $k_r$ are uniformly bounded from above. Therefore the jump rates $k\super{V}$ are also bounded and so there are almost surely a finite number of jumps and the process does not explode. However, the crucial challenge of \cite{AAPR2021TR} and previous works on this large-deviation principle is that the rates $k_r$ are generally not lower-bounded away from zero. The condition $\theta^->0$ assures that zero temperature, where the Arrhenius factors vanish $A_r(0)=0$, can never be reached. By contrast, concentrations for which $B_r(\rho)=0$ are not problematic in this respect, since the mass-action factors $B_r(\rho)$ vanish sufficiently slow. We explain this through two simple examples.

\begin{example}[nonreversible cell division]
In this case $\X\coloneqq \{\mathsf{A}\}, \R\coloneqq \{\mathds1_\mathsf{A}\to2\mathds1_\mathsf{A}\}$ and $\se_\mathsf{A}\coloneqq 0$ so that the temperature is kept constant and the condition $\theta^->0$ is superfluous. Starting from $\rho^0_\mathsf{A}\coloneqq 0$, $\theta^0\coloneqq 1$, the solution to the macroscopic equation~\eqref{eq:ODE} will simply be $\rho_{\mathsf{A}}(t)\equiv0$. However, the macroscopic initial condition $\rho^0_\mathsf{A}=0$ would be consistent with the microscopic condition $\rho\supernul{V}_{\mathsf{A}}=1/V$, meaning that initially there is one particle, which is not observed on the macroscopic scale $V\to\infty$. Although not the expected behaviour, there is a small probability that this one particle causes a chain reaction, resulting in the emergence -- within finite time --  of a concentration $\bigoh(\rho\super{V}_{\mathsf{A}}(1))=1$ of macroscopic order. Since $B_r(\rho)$ grows polynomially, it turns out that this probability has finite large-deviation cost~\cite[Sec.~3]{AAPR2021TR}: there exists a path $\rho_\mathsf{A}(t)$ with $\rho_\mathsf{A}(0)=\rho^0_\mathsf{A}=0$ and $\rho_\mathsf{A}(T)>0$ such that $\I_{\lbrack0,T\rbrack}(\rho)<\infty$. Physically, this means that the process can `escape' the boundary $\rho_{\mathsf{A}}=0$ if a sufficient amount of free energy is injected in the system. 
\label{ex:cell division}
\end{example}

\begin{example}[heating the room]
We consider again one species $\X\coloneqq \{\mathsf{A}\}$, but now $\R\coloneqq \{\mathds1_\mathsf{A}\to0\}$ and $\se_\mathsf{A}>0$. Initially we set $\rho^0_\mathsf{A}=\rho\super{V}_\mathsf{A}(0)\coloneqq 1$, $\theta^0\coloneqq 0$ and $\Theta\super{V}(0)=\theta\supernul{V}\coloneqq a/V$ for some $a>0$ arbitrarily large but of order $1$. Similar to the previous example, the solution to the macroscopic equation~\eqref{eq:ODE} is constant $(\rho_\mathsf{A}(t),\theta(t))\equiv(0,0)$, but there is some heat available to start the reaction, increase the temperature and thereby accelerate the process (at least as long as there is enough mass available). However, in this case it can be calculated explicitly that the probability of a sufficiently strong chain reaction corresponds to an infinite large-deviation cost: for any path $(\rho_\mathsf{A},\theta(t))$ with $(\rho_{\mathsf{A}}(0),\theta(0))=(\rho^0_\mathsf{A},\theta^0)=(1,0)$ and $\theta(T)>0$ there holds $\I_{\lbrack0,T\rbrack}(\rho)<\infty$ \cite[Sec.~5]{AAPR2021TR}. Physically, this simple model shows that a room at $0$ Kelvin can not be heated with fire, since there not enough heat to start the fire, even when injecting free energy in the system. More precisely, the enormous amount of free energy that would need to be injected in to start the fire is of a higher scale. 
\label{ex:heating}
\end{example}

Why can the system of the first example escape the crucial boundary with finite large-deviation cost, but the second system cannot? The necessary and sufficient condition to escape a boundary point with large-deviation cost is that \cite{AAPR2021TR}:
\begin{equation}
  \lim_{\tau\to0} -\int_0^{\tau}\!\log k_r(\rho^\mathrm{bd}+\tilde\tau g)\,d\tilde\tau=0,
\label{eq:AAPR condition}
\end{equation}
where $\rho^\mathrm{bd}$ is a boundary point of $\{\rho\in\fS:k_r(\rho)>0\}$ and $g$ is an inward-pointing vector.

For the first example, $\rho^\mathrm{bd}_\mathsf{A}\coloneqq 0$, and we can simply take $g\coloneqq \mathds1_\mathsf{A}$. Focussing on the crucial factor $B_r$ of the reaction rate $k_r$, we see that indeed:
\begin{equation*}
  -\int_0^{\tau}\!\log B_r(\rho^\mathrm{bd}+\tilde\tau g)\,d\tilde\tau=-\int_0^{\tau}\!\log (\tilde\tau)\,d\tilde\tau= -\tau\log\tau+\tau\to0,
\end{equation*}
For the second example $\rho_\mathsf{A}^\mathrm{bd}\coloneqq E^0/\se_\mathsf{A}$ and for an inward pointing vector we can take $g\coloneqq -c_H/\se_\mathsf{A}\mathds1_\mathsf{A}$. Then, again focussing on the crucial factor $A_r$ of $k_r$:
\begin{equation*}
  -\int_0^{\tau}\!\log A_r(\rho^\mathrm{bd}+\tilde\tau g)\,d\tilde\tau=\frac{c_H(\sa-\se_\mathsf{A})}{k_B\se_\mathsf{A}} \int_0^{\tau}\!\frac1{\tilde\tau}\,d\tilde\tau\equiv\pm\infty,
\end{equation*}
and so the reaction rate vanishes too fast near the boundary in order to escape with finite cost. It is interesting to note that in this respect the Arrhenius law is exactly a border case which does not allow escaping the boundary $\theta=0$.

In order to circumvent this issue we assume $\theta^->0$ in Theorem~\ref{th:LDP conc}.

\subsection{Flux large deviations}
\label{subsec:flux ldp}

In order to generalise the physical structure beyond detailed balance, we also study large deviations of fluxes. At this stage we require reversibility of the reaction network (each forward reaction corresponds to a backward reaction), permitting the use of net fluxes. To be more precise, let
\begin{align*}
  W\super{V}_r(t)&\coloneqq  \frac1V\#\big\{\text{reactions } r \text{ occurred in } (0,t\rbrack\big\} \\
  &\qquad - \frac1V\#\big\{\text{reactions } \bw(r) \text{ occurred in } (0,t\rbrack\big\}, \qquad r\in\R_\fw.
\end{align*}
be the cumulative net reaction flux. 

As mentioned in Subsection~\ref{subsec:cold death}, there are almost surely a finite number of jumps. Therefore, the paths $W\super{V}$ are almost surely of bounded variation, and we may define the time derivative $J\super{V}\coloneqq \dot W\super{V}$, which is the flux from Subsection~\ref{subsec:goal}. This flux is however a singular measure in time, so for now it is more convenient to work with the cumulative net flux $W\super{V}(t)$, which is a Markov process in $\RR^\R$ with initial condition $W\super{V}(0)\equiv0$ and generator:
\begin{multline*}
  (\Q\super{V}F)(w) = V\sum_{r\in\R_\fw}  \Big\lbrack \kappa_r A_r(\theta) B\super{V}_r(\rho) \big\lbrack F(w+\tfrac1V\mathds1_r)-F(w)\big\rbrack\\
  +\kappa_{\bw(r)} A_{\bw(r)}(\theta) B\super{V}_{\bw(r)}(\rho) \big\lbrack F(w-\tfrac1V\mathds1_r)-F(w)\big\rbrack \Big\rbrack,
\end{multline*}
where the cumulative net flux $w$ determines the concentration through the \emph{continuity equation} $\rho  \coloneqq   \rho\supernul{V} + \Gamma w$ (recall the definition of $\Gamma$ from \eqref{eq:Gamma}), and again $\theta \coloneqq  (E\supernul{V}-\se\cdot\rho)/c_H$ . On an abstract level one then sees that the condition~\eqref{eq:AAPR condition} is the same as for the process~$\rho\super{V}$, yielding the flux large-deviation principle:

\begin{theorem}[\cite{PR2019,AAPR2021TR}]
Assume $\theta^0,\rho^0$ and $\R$ are such that $\fS$ is bounded and $\theta^->0$. Then the process $W\super{V}$ satisfies a large-deviation principle in $D(0,T;\RR^\X)$ with rate functional
\begin{align}
  \J_{\lbrack0,T\rbrack}(w)&\coloneqq 
  \begin{cases}
    \int_0^T\!\L\big(\rho(t),\dot w(t)\big)\,dt, & w \in W^{1,1}(0,T;\RR^{\R_\fw}), \quad w(0)=0,\\
                                                  &\rho\coloneqq \rho^0+\Gamma w,\\
    \infty, &\text{otherwise}, 
  \end{cases}\label{eq:J}
\intertext{where}
  \L\big(\rho,j)&\coloneqq \inf_{\substack{\tilde\jmath\in\lbrack0,\infty)^\R: \\ \forall r\in\R_\fw\quad j_r=\tilde\jmath_r-\tilde\jmath_{\bw(r)}}} \quad \S\big(\tilde\jmath\mid k(\rho,\theta)\big).
  \label{eq:L}
\end{align}
\label{th:LDP flux}
\end{theorem}
The infimum in \eqref{eq:L} appears by a ``contraction principle'' since we work with net rather then one-way fluxes. By the same principle, \eqref{eq:hatL} is related to \eqref{eq:L} through 
\begin{equation}
  \hat\L(\rho,u)=\inf_{j\in\R_\fw: u=\Gamma j} \L(\rho,j).
\label{eq:contracted L}
\end{equation}


\section{Invariant measure and quasipotential}
\label{sec:QP}

This section is dedicated to the invariant measure $\Pi\super{V}$ of the microscopic process, and its large-deviation rate $\V$, also called the quasipotential\footnote{Actually, the most common notion of quasipotential as used in \cite{Freidlin2012} is slightly different. There are known cases where that notion does not coincide with the large-deviation rate of the invariant measure, even if the macroscopic dynamics has a unique basin of attraction~\cite{YasodharanSundaresan2021TR}.}.
For a simple unimolecular reaction network $\mathsf{A}\rightleftharpoons\mathsf{B}$ the invariant measure will be explicitly constructed in Subsection~\ref{subsec:unim}, and its large-deviation rate is calculated in Subsection~\ref{subsec:QP unim}. For more general networks, the invariant measure remains unknown, but the quasipotential can still be derived indirectly via a Hamilton-Jacobi-Bellman equation; this will be the content of Subsection~\ref{subsec:QP general}.

\subsection{The invariant measure for the unimolecular case}
\label{subsec:unim}

In this subsection we explicitly construct the invariant measure for the simple setting $\X\coloneqq \{\mathsf{A},\mathsf{B}\}$ and $\R\coloneqq \{\fw,\bw\}\coloneqq \{\mathds1_\mathsf{A}\to\mathds1_\mathsf{B},\mathds1_\mathsf{B}\to\mathds1_\mathsf{A}\}$, in chemical notation: $\mathsf{A}\rightleftharpoons\mathsf{B}$. To simplify notation we write $\sa\coloneqq \sa_\fw=\sa_\bw$ and without loss of generality we assume that $V$ is the total number of particles in the system (conserved by the dynamics). It will be helpful to rewrite all variables in terms of the number $i$ of $\mathsf{A}$-particles (for fixed V):
\begin{align*}
  \rho\lbrack i\rbrack\coloneqq \begin{bmatrix} \tfrac1V i \\ 1-\tfrac1Vi \end{bmatrix},
  &&
  \theta\lbrack i\rbrack\coloneqq \frac{E\supernul{V}-\se\cdot\rho\lbrack i\rbrack}{c_H} 
  = 
  \frac{E\supernul{V}-\se_\mathsf{B}}{c_H} + \frac{\se_\mathsf{B}- \se_\mathsf{A}}{c_H V}i,
\end{align*}
recalling the total energy defined in~\eqref{eq:total energy micro}. Define $i^-\coloneqq \min\{i=0,\hdots,V:\theta\lbrack i\rbrack\in \lbrack\theta^-,\theta^+\rbrack\}$ and similarly $i^+$ as the maximum. Then the process $(\rho\super{V},\Theta\super{V})$ can be interpreted as a birth-death process on the state space $\{i^-,\hdots,i^+\}$. 

For a birth-death process the invariant measure $\Pi\super{V}$ can be explicitly constructed using detailed balance as Ansatz:
\begin{align}
  \Pi\super{V}(\rho\lbrack i\rbrack) &= \Pi\super{V}(\rho\lbrack i-1\rbrack)\frac{k\super{V}_\bw(\rho\lbrack i-1\rbrack,\theta\lbrack i-1\rbrack)}{k\super{V}_\fw(\rho\lbrack i\rbrack,\theta\lbrack i\rbrack)}  \notag\\
    &= \Pi\super{V}(\rho\lbrack i^-\rbrack) \frac{\prod_{l=i^-}^{i-1} \kappa_\bw A_\bw(\theta\lbrack l\rbrack)B_\bw(\rho\lbrack l\rbrack)}{\prod_{l=i^-+1}^i\kappa_\fw A_\fw(\theta\lbrack l\rbrack) B_\fw(\rho\lbrack l\rbrack)}.
\label{eq:unim DB Ansatz}
\end{align}
From this we see that the measure factorises into chemical and thermal factors, so we can set
\begin{align}
  \Pi\super{V}(\rho\lbrack i\rbrack)=:\frac{F\super{V}(\theta\lbrack i\rbrack)G\super{V}(\rho\lbrack i\rbrack)}{Z\super{V}},
  &&
  Z\super{V}\coloneqq \sum_{i=i^-}^{i^+} F\super{V}(\theta\lbrack i\rbrack) G\super{V}(\rho\lbrack i\rbrack),
  \label{eq:unim inv meas}
\end{align}
for some $F\super{V}, G\super{V}$ that can be studied separately. Note that $\Pi\super{V}(\rho\lbrack i^-\rbrack)$ and possible other constants that we encounter can be absorbed in the normalisation factor $Z\super{V}$. 

The chemical factor has the Poisson-product form that is the invariant measure for the isothermal reaction network (see for example~\cite{AndersonCraciunKurtz2010}):
\begin{align*}
  \frac{\prod_{l=i^-}^{i-1} \kappa_\bw B_\bw(\rho\lbrack l\rbrack)}{\prod_{l=i^-+1}^i\kappa_\fw B_\fw(\rho\lbrack l\rbrack)} 
  &=\frac{\prod_{l=i^-}^{i-1} \kappa_\bw (v-l)}{\prod_{l=i^-+1}^i \kappa_\fw l} \\
  &=\big(\mfrac{\kappa_\bw}{\kappa_\fw}\big)^{i-i^-} \frac{(V-i^-)!/(V-i)!}{i!/i^-!} \\
  &=\text{constant}\times \underbrace{\prod_{x=\mathsf{A},\mathsf{B}}\frac{(V\pi_x)^{V\rho_x\lbrack i\rbrack}}{(V\rho_x\lbrack i\rbrack)!} e^{-V\pi_x}}_{=:G\super{V}(\rho\lbrack i \rbrack)}.
\end{align*}
setting $\pi\coloneqq (\kappa_\bw,\kappa_\fw)/\lvert(\kappa_\bw,\kappa_\fw)\rvert_2$.

For the thermal factor, we obtain:
\begin{align*}
  \frac{\prod_{l=i^-}^{i-1} A_\bw(\theta\lbrack l\rbrack)}{\prod_{l=i^-+1}^i A_\fw(\theta\lbrack l\rbrack)} &= 
  \frac{A_\bw(\theta\lbrack i^-\rbrack)}{A_\fw(\theta\lbrack i\rbrack)}
    \prod_{l=i^-+1}^{i-1}  \frac{ A_\bw(\theta\lbrack l\rbrack)}{A_\fw(\theta\lbrack l\rbrack)}\\
  &= \exp\Big(-\frac{\sa-\se_\mathsf{B}}{k_B\theta\lbrack i^-\rbrack}\Big)   \exp\Big(\frac{\sa-\se_\mathsf{A}}{k_B\theta\lbrack i\rbrack}\Big) \exp\Big( \sum_{l=i^-+1}^{i-1} \frac{\se_\mathsf{B}-\se_\mathsf{A}}{k_B\theta\lbrack l\rbrack} \Big)\\
  &=: F\super{V}(\theta\lbrack i\rbrack).
\end{align*}
By construction \eqref{eq:unim inv meas} satisfies the Ansatz~\eqref{eq:unim DB Ansatz} and hence the process is in detailed balance with respect to the invariant measure $\Pi\super{V}$.

Motivated by the results in Subsection~\ref{subsec:QP general}, we expect a similar formula for the invariant measure to hold more generally when the isothermal reaction network is in detailed balance, or if it is in complex balance and $\sa$ is constant.

\subsection{The quasipotential for the unimolecular case}
\label{subsec:QP unim}

Still restricting to the simple setting $\mathsf{A}\rightleftharpoons\mathsf{B}$ we derive the quasipotential $\V$ as the large-deviation rate function corresponding to the invariant measure $\Pi\super{V}$ constructed in the previous subsection. As explained in \cite[Sec.~2]{MielkePattersonPeletierRenger2016} it corresponds to the (nondimensionalised) physical free energy per unit volume. Moreover, the function $\V$ will generally be a Lyapunov function for the macroscopic dynamics, and -- in case the Onsager-Machlup principle holds -- also the driving energy, see Sections~\ref{sec:Onsager-Machlup} and ~\ref{sec:MFT}.

Since the large-deviation principle~\eqref{eq:LDP Pi} is a finite-dimensional problem we present a formal but direct calculation and skip the minor technicalities required to make this into a rigorous statement. Pick a concentration-temperature pair $(\rho,\theta)\in\lbrack0,\infty)^\X\times\lbrack0,\infty)$ and a sequence $(i\super{V})_{V>0}$ in $\{i^-,\hdots,i^+\}$ so that $(\rho\lbrack i\super{V}\rbrack,\theta\lbrack i\super{V}\rbrack)\to(\rho,\theta)$ as $V\to\infty$, adopting the notation of the previous subsection. Then
\begin{multline*}
  -\frac1V\log\Pi\super{V}\big( \rho\lbrack i\super{V}\rbrack,\theta\lbrack i\super{V}\rbrack\big) \stackrel{\eqref{eq:unim inv meas}}{=} \\ 
  -\frac1V\log F\super{V}\big(\theta\lbrack i\super{V}\rbrack\big) -\frac1V\log G\super{V}\big(\rho\lbrack i\super{V}\rbrack\big) + \frac1V\log Z\super{V}.
\end{multline*}
It is well known and easily seen by Stirling's formula that
\begin{equation*}
  \lim_{V\to\infty} -\frac1V\log G\super{V}\big(\rho\lbrack i\super{V}\rbrack\big) = \S(\rho\mid\pi).
\end{equation*}
For the thermal contribution note that $i^-$ and $\theta\lbrack\cdot\rbrack$ defined in the previous subsection depend on $V$ and that $\theta\lbrack i^-\rbrack\to\theta^-$ as $V\to\infty$. By a Riemann integral approximation:
\begin{align*}
  \lim_{V\to\infty} -\frac1V\log F\super{V}\big(\theta\lbrack i\super{V}\rbrack\big) 
  & = \lim_{V\to\infty} -\frac1V \sum_{l=i^-+1}^{i-1} \frac{\se_\mathsf{B}-\se_\mathsf{A}}{k_B\theta\lbrack l\rbrack} \\
  & = \lim_{V\to\infty} -\frac1V \sum_{l=i^-+1}^{i-1} \frac{\se_\mathsf{B}-\se_\mathsf{A}}{k_B\big( \frac{E^0-\se_\mathsf{B}}{c_H} + \frac{\se_\mathsf{B}- \se_\mathsf{A}}{c_H }\rho\lbrack l\rbrack\big)} \\
  & = - \frac{\se_\mathsf{B}-\se_\mathsf{A}}{k_B} \int_{\lim_{V\to\infty}\rho\lbrack i^-\rbrack}^{\rho}\!\frac{1}{\frac{E^0-\se_\mathsf{B}}{c_H} + \frac{\se_\mathsf{B}- \se_\mathsf{A}}{c_H }\tilde\rho}\,d\tilde\rho \\
  & = - \frac{c_H}{k_B} \log \frac{\theta}{\theta^-}.
\end{align*}
Putting all parts together indeed yields the claimed quasipotential~\eqref{eq:QP}.

\subsection{The quasipotential for the general case}
\label{subsec:QP general}

For the more general case as the simple unimolecular example, it still remains unknown what the exact expression of the invariant measure is, and therefore the direct calculation of its large-deviation rate functional as in the previous subsection cannot be used. However, we can still derive the quasipotential indirectly, without knowing the invariant measure. The key observation is that the invariant measure $\Pi\super{V}$ is related to the generator $\Q\super{V}$, and so the corresponding large-deviation rate $\V$ must be related to the large-deviation cost function $\hat\L$. We recall the exact relation without proof.
\begin{lemma}[{\cite[Th.~3.6]{PattersonRengerSharma2021TR}}]
If the invariant measure $\Pi\super{V}$ satisfies a large-deviation principle~\eqref{eq:QP} with rate function $\V$, then
\begin{align}
  \hat\H\big(\rho,\grad\V(\rho)\big)=0, && \inf\V=0,
  \label{eq:HJB}
\end{align}
where $\hat\H$ is the convex dual of $\hat\L$:
\begin{equation}
  \hat\H\big(\rho,\xi)\coloneqq \sup_{u\in\RR^\X} \xi\cdot u - \hat\L\big(\rho,u) = \sum_{r\in\R} k_r(\rho)\big( e^{\xi\cdot\gamma\super{r}}-1\big).
\label{eq:hatH}
\end{equation}
\label{lem:HJB}
\end{lemma}

From now on we focus on equation~\eqref{eq:HJB}, which is exactly what is needed for the use in Onsager-Machlup and Macroscopic Fluctuation Theory. In order to solve it, we will need to make assumptions about the isothermal setting, where $A_r(\theta)\equiv1$. For that setting, the \emph{stoichiometric simplex} captures all concentrations that can be reached through reactions in $\R$:
\begin{equation*}
  \fS^\iso\coloneqq \big\{ \rho=\rho^0 + \Gamma w: w\in \RR^{\R_\fw} \text{ and } \rho \in \lbrack0,\infty)^\X \big\}.
\end{equation*}
Note that $\fS \subset \fS^\iso$ since some concentrations in $\fS^\iso$ could correspond to negative temperatures. Either one of the following two assumptions are needed:
\begin{quote}
  \textbf{Isothermal Detailed Balance (IDB):} There exists a unique state $\pi\in\fS^\iso$ such that for each forward reaction $r\in\R_\fw$:
  \begin{equation*}
    \kappa_r B_r(\pi) = \kappa_{\bw(r)} B_{\bw(r)}(\pi).
  \end{equation*}
\end{quote}

\begin{quote}
  \textbf{Isothermal Complex Balance (ICB):}   There exists a unique state $\pi\in\fS^\iso$ such that for each complex $\alpha\in \{\alpha\super{r}:r\in\R\}$:
  \begin{equation*}
    \sum_{r\in\R:\alpha\super{r}=\alpha} \kappa_r B_r(\pi) = \sum_{r\in\R:\alpha^{\bw(r)}=\alpha} \kappa_{r} B_r(\pi).
  \end{equation*}
\end{quote}
Naturally for both assumptions $\pi$ is the steady state for the macroscopic isothermal equation $\dot\rho(t)=\Gamma j(t), \,j_r(t)=\kappa_r B_r(\rho(t)) - \kappa_{\bw(r)} B_{\bw(r)}(\rho(t))$.

A direct calculation shows that (ICB) is equivalent to~\cite[Sec.~3.2]{AndersonCraciunKurtz2010},
\begin{equation}
  \sum_{r\in\R_\fw} \big(  \kappa_r B_r(\pi) -\kappa_{\bw(r)}B_{\bw(r)}(\pi) \big)
  \big(\psi_{\alpha^{\bw(r)}} - \psi_{\alpha\super{r}}\big) =0
\label{eq:ICB2}  
\end{equation}
for all test functions $\psi$ mapping complexes $\alpha\in\{\alpha\super{r}:r\in\R\}$ to the reals. This will be useful in the next proof.

\begin{remark}
Both conditions (IDB) and (ICB) are macroscopic in nature. However, since our network is assumed to be reversible, (IDB) is equivalent to detailed balance of the isothermal microscopic system, in the sense of reversibility of the Markov process~\cite[Th.~4.5]{AndersonCraciunKurtz2010}.
\end{remark}


We now come to our first main result.

\begin{prop} Assume that either (IDB) holds, or that (ICB) holds and $\sa_r$ is constant in $r\in\R$, and let $\pi\in\lbrack0,\infty)^\X$ be the corresponding steady state of the isothermal equation. Then at all points of differentiability of $\V$, the Hamilton-Jacobi-Bellman equation~\eqref{eq:HJB} holds with
\begin{align}
  \V(\rho)\coloneqq \S(\rho\mid\pi) - \frac{c_H}{k_B}\log\theta - C, && \theta\coloneqq \frac{E^0-\se\cdot\rho}{c_H},
\label{eq:QP2}
\end{align}
where
\begin{equation}
  C=\inf_{\substack{\rho\in\fS,\\ \theta\coloneqq (E^0-\se\cdot\rho)/c_H}} \S(\rho\mid\pi)-\mfrac{c_H}{k_B} \log \theta.
\label{eq:QP normalisation}
\end{equation}
\label{prop:QP}
\end{prop}

\begin{proof}
We first focus on the gradient of $\V$ and comment on the constant $C$ at the end. Under (IDB) or (ICB), the steady state $\pi$ is coordinate-wise positive~\cite[Th.~3.2]{AndersonCraciunKurtz2010}. Thus, abbreviating $\log(\rho/\pi)\coloneqq (\log(\rho_x/\pi_x))_{x\in\X}$,
\begin{align}
  \grad\V(\rho)=\log\frac{\rho}{\pi} + \frac{\se }{k_B\theta}, && \theta\coloneqq \frac{E^0-\se\cdot\rho}{c_H},
\label{eq:gradV}
\end{align}
at all points of differentiability $\{\rho>0:\se\cdot\rho<E^0\}$.
Plugging this gradient into $\hat\H$:
\begin{align*}
  \hat\H\big(\rho,\grad\V(\rho)\big)
  &=\theta^q \sum_{r\in\R} \kappa_r e^{-\frac{\sa_r - \se\cdot\alpha\super{r}}{k_B\theta}}\rho^{\alpha\super{r}}
  \Big\lbrack e^{
      (\log\frac{\rho}{\pi}+\frac{\se}{k_B\theta})\cdot\gamma\super{r}
    } -1
  \Big\rbrack\\
  &=\theta^q \sum_{r\in\R_\fw} \big(  \kappa_r \pi^{\alpha\super{r}}-\kappa_{\bw(r)}\pi^{\alpha^{\bw(r)}} \big)\\
  &\hspace{5em}\times
  \Big( 
    \big(\tfrac{\rho}{\pi}\big)^{\alpha^{\bw(r)}} e^{-\frac{\sa_r-\se\cdot \alpha^{\bw(r)}}{k_B\theta}}
    -
    \big(\tfrac{\rho}{\pi}\big)^{\alpha\super{r}} e^{-\frac{\sa_r-\se\cdot \alpha\super{r}}{k_B\theta}}
  \Big).
\end{align*}
If (IDB) holds, then clearly the first bracket is zero for all $r\in\R_\fw$.

If $\sa_r=0$, then all variables in the second expression depend on the complex only. Thus
\begin{align*}
  \hat\H\big(\rho,\theta,\grad\V(\rho)\big)
  &=\theta^q \sum_{r\in\R_\fw} \big(  \kappa_r \pi^{\alpha\super{r}}-\kappa_{\bw(r)}\pi^{\alpha^{\bw(r)}} \big)
  \big(\psi_{\alpha^{\bw(r)}} - \psi_{\alpha\super{r}}\big),
\end{align*}
where the test function $\psi:\{\alpha\super{r}:r\in\R\}\to\RR$ is defined as $\psi_\alpha\coloneqq (\rho/\pi)^{\alpha}\exp(\se\cdot\alpha/(k_B\theta))$. This expression is zero under the assumption (ICB) because of~\eqref{eq:ICB2}.

It remains to show that $\inf\V=0$, which is clearly true whenever $C$ is finite. This is easily seen by recalling that
\begin{equation}
  \se\cdot\rho+c_H\theta=\se\cdot\rho^0+c_H\theta^0=:E^0,
\label{eq:energy conservation}
\end{equation}
and so one can bound from below:
\begin{equation}
  \S(\rho\mid\pi) - \frac{c_H}{k_B}\log\theta \geq \frac{c_H}{k_B}(1-\theta)
   = \frac{c_H}{k_B}(1-\frac{E^0 -\se\cdot\rho}{c_H})
  \geq \frac{c_H-E^0}{k_B}.
  \label{eq:V bdd below}
\end{equation}

\end{proof}
\begin{remark} The proof shows that the factor $\theta^q$ plays no role whatsoever since it is independent of the reaction $r$. Moreover, the proof will generally break down if (ICB) holds but $\sa_r$ is not constant in $r$.
\end{remark}

Now that the quasipotential for the anisothermal dynamics is known and is a strictly convex Lyaponuv functional for the macroscopic dynamics~\eqref{eq:ODE}, the steady state $(\rho^\infty,\theta^\infty)$ can be found as the minimiser in the minimisation problem~\eqref{eq:QP normalisation}. Since $\V$ is a strictly convex function on a finite-dimensional space that is bounded from below by~\eqref{eq:V bdd below} and $\fS$ is convex, the minimiser exists and is unique. Let $P$ be the projection from $\RR^\X$ onto $\mathop{\mathrm{Ran}}\Gamma$:
\begin{equation*}
  P_{xy} \coloneqq  \sum_{r\in\R_\fw} \frac{\alpha\super{r}_x \alpha\super{r}_y}{\lvert \alpha\super{r}\rvert_2^2}.
\end{equation*}
Differentiating $\V$ yields that the minimiser satisfies 
\begin{align*}
  \rho^\infty_x = \pi_x e^{-\frac{\se_x}{k_B\theta^\infty}+((I-P)\lambda)_x}, &&
  \se\cdot\rho^\infty + c_H\theta^\infty = E^0,
\end{align*}
where the Lagrange multiplier $\lambda\in\RR^\X$ is chosen such that $(P-I)(\rho^\infty-\rho^0)=0$. It is difficult to obtain a more explicit expression, even without the constraint $\rho\in\rho^0+\mathop{\mathrm{Ran}}\Gamma$. 

\begin{remark}
  The quasipotential can now also be written as:
\begin{equation*}
  \V(\rho)=\S(\rho\mid\pi) -\S(\rho\mid\rho^\infty) - \frac{c_H}{k_B}\log\frac{\theta}{\theta^\infty}.
\end{equation*}
\end{remark}

\section{Onsager-Machlup Theory}
\label{sec:Onsager-Machlup}

In \cite{MielkePeletierRenger2014} we showed a modern version of the Onsager-Machlup principle. Translated to the setting of this paper, this means that if the Markov process $\rho\super{V}(t)$ is in (stochastic) detailed balance with respect to its invariant measure $\Pi\super{V}$ that has large-deviation rate $\V$, then there exist a unique \emph{dissipation potential} $\hat\Psi:\lbrack0,\infty)^\X\times\lbrack0,\infty)\times \RR^\X\times\RR$ so that
\begin{equation}
  \hat\L\big(\rho,u\big) = \hat\Psi\big(\rho,u\big) + \hat\Psi^*\big(\rho,-\tfrac12\grad\V(\rho)\big)
  +\tfrac12\grad\V(\rho)\cdot u
\label{eq:hatL decomp}
\end{equation}
for all $u\in\RR^\X$ and all points of differentiability $\rho\in\lbrack0,\infty)^\X$ of $\V$.
By definition, $\hat\Psi:\lbrack0,\infty)^\X\times\RR^\X\to\lbrack0,\infty)$ being a dissipation potential means that $\hat\Psi(\rho,u)$ is convex in $u$ and $\hat\Psi(\rho,0)\equiv0$, which reflects the physical principle that there is no dissipation at zero velocities. Of course, \eqref{eq:hatL decomp} is the local version of~\eqref{eq:Onsager-Machlup}.

It is difficult to check whether $\rho\super{V}(t)$ is in detailed balance if the invariant measure is not known. Luckily, detailed balance only needs to hold in the large-deviation regime. On the stochastic microscopic level, detailed balance with respect to an invariant measure $\Pi\super{V}$ means that:
\begin{align*}
  \Prob\big(\rho\super{V} \in d\rho\big) = \Prob\big(\overleftarrow\rho\super{V}\in d\rho\big), && \overleftarrow\rho\super{V}(t)\coloneqq \rho\super{V}(T-t),
\end{align*}
whenever $\rho\super{V}(0)$ is distributed according to $\Pi\super{V}$. Taking the large deviations on both sides $\V(\rho(0))+\int_0^T\!\hat\L(\rho(t),\dot\rho(t))\,dt=\V(\rho(T))+\int_0^T\!\hat\L(\rho(t),-\dot\rho(t))\,dt$, and since this holds for all $T$, one finds the following time-reversal symmetry on the large-deviations scale (see~\cite{MielkePeletierRenger2014}):
\begin{equation}
  \hat\L(\rho,u)=\hat\L(\rho,-u) + \grad\V(\rho)\cdot u
\label{eq:time-reversal symm1}
\end{equation}
for all points of differentiability $\rho$ of $\V$ and all $u\in\RR^\X$.
\begin{prop}
Assume (IDB), so that by Proposition~\ref{prop:QP} the quasipotential $\V$ is given by \eqref{eq:QP2}. Then the time-reversal symmetry~\eqref{eq:time-reversal symm1} holds.
\label{prop:DB time-reversal symm}
\end{prop}

\begin{proof}
Taking the convex dual on both sides shows that \eqref{eq:time-reversal symm1} is equivalent to the symmetry of $\H(\rho,\xi+\tfrac12\grad\V(\rho))$ in $\xi$. We thus calculate, abbreviating $\theta=(E^0-\se\cdot\rho)/c_H$ and using~\eqref{eq:gradV}:
\begin{align*}
  &\H\big(\rho,\xi+\tfrac12\grad\V(\rho)\big) - \H\big(\rho,-\xi+\tfrac12\grad\V(\rho)\big)\\
  &\qquad =\sum_{r\in\R} \kappa_r A_r(\theta) B_r(\rho) 
    \Big( e^{(\frac12\log\frac\rho\pi + \frac12\frac{\se}{k_B\theta}+\xi)\cdot\gamma\super{r}}-e^{(\frac12\log\frac\rho\pi + \frac12\frac{\se}{k_B\theta}-\xi)\cdot\gamma\super{r}}\Big)\\
  &\qquad =\sum_{r\in\R} \kappa_r \pi^{-\frac12\gamma\super{r}} e^{-\frac{\sa_r-\se\cdot(\alpha\super{r}+\alpha^{\bw(r)})/2}{k_B\theta}} \rho^{\frac12(\alpha\super{r}+\alpha^{\bw(r)})}
    \Big(e^{\xi\cdot\gamma\super{r}}-e^{-\xi\cdot\gamma\super{r}}\Big) \\
  &\qquad=
  \sum_{r\in\R} \Big( \kappa_r \pi^{-\frac12\gamma\super{r}} - \kappa_{\bw(r)} \pi^{\frac12\gamma\super{r}}\Big) e^{-\frac{\sa_r-\se\cdot(\alpha\super{r}+\alpha^{\bw(r)})/2}{k_B\theta}} \rho^{\frac12(\alpha\super{r}+\alpha^{\bw(r)})}
    e^{\xi\cdot\gamma\super{r}}.
\end{align*}
The last line follows from $e^{-\xi\cdot\gamma\super{r}}=e^{\xi\cdot\gamma^{\bw(r)}}$ and changing the index. The first factor is zero due to (IDB):
\begin{equation*}
  \kappa_r \pi^{-\frac12\gamma\super{r}} - \kappa_{\bw(r)} \pi^{\frac12\gamma\super{r}} = \pi^{-\frac12(\alpha\super{r}+\alpha^{\bw(r)})} \big( \kappa_r\pi^{\alpha\super{r}} - \kappa_{\bw(r)}\pi^{\alpha^{\bw(r)}}\big)=0.
\end{equation*}
\end{proof}

Under the (IDB) assumption, as an immediate consequence of the time-reversal symmetry~\eqref{eq:time-reversal symm1} and \cite[Th.~2.1]{MielkePeletierRenger2014}, the Onsager-Machlup decomposition~\eqref{eq:hatL decomp} holds, with the uniquely given dissipation potentials:
\begin{align*}
  \hat\Psi^*(\rho,\xi)&= 2 \sum_{r\in\R_\fw} \sqrt{k_r(\rho)k_{\bw(r)}} \big(\cosh(\xi\cdot\gamma_r)-1\big),\\
  \hat\Psi(\rho,u)&=\sup_{\xi\in\RR^\X} \xi\cdot u - \hat\Psi^*(\rho,\xi).
\end{align*}
This result can be seen as an extension of \cite{MielkePattersonPeletierRenger2016} to include temperature effects. In that paper it is explained how $\hat\Psi$ can be written as an infimal convolution, which also applies to our anisothermal setting.

\begin{remark}
In general $\hat\Psi$ does not have a clean explicit formulation (apart from the infimal convolution). However, it does have the dual formulation:
\begin{equation}
  \hat\Psi(\rho,u)=\inf_{\substack{j\in\RR^{\R_\fw}:\\u=\Gamma j}} \, \sum_{r\in\R_\fw} 2\sqrt{k_r(\rho)k_{\bw(r)}(\rho)}  \Big(\cosh^*\big(\mfrac{j_r}{2\eta_r(\rho)}\big)+1\Big).
\label{eq:hatPsi}
\end{equation}
The appearance of the constrained infimum is no coincidence: $\hat\Psi$ is related to the dissipation potential for the MFT setting in the same way as \eqref{eq:contracted L}, as was shown in \cite{Renger2018}.
\label{rem:hatPsi}
\end{remark}

\section{Macroscopic Fluctuation Theory}
\label{sec:MFT}

In this section we move away from (isothermal) detailed balance (IDB) to more general dynamics. Physically, this means that we look beyond thermodynamically closed systems, allowing for interactions with the environment. However, since we will need the explicit expression of the quasipotential $\V$, we need to assume (ICB) and that $\sa_r$ is constant in $r$. 

We first derive the first decomposition~\eqref{eq:L decomp1} and introduce the symmetric and asymmetric forces. Then the rate functional is further decomposed using generalised orthogonality, and we end with an interpretation of the resulting MFT decomposition.

\subsection{Thermodynamic force and dissipation}

One of the aims of MFT is to uniquely distinguish between internal mechanisms, that would be present if the system was thermodynamically closed, and mechanisms that should be regarded as interactions with the environment. Naively one could look for a decomposition of the type~\eqref{eq:hatL decomp}, where the force $-\tfrac12\grad\V$ is replaced by a more general, non-conservative force. However, it is not clear whether such force indeed exists; this is one mathematical motivation to switch to the flux setting. Another, more physical motivation is the observation that systems that are driven out of equilibrium by external forces can potentially give rise to divergence-free cycles\footnote{In this context, the operator $\Gamma$ plays the role of a divergence, so divergence-free fluxes are fluxes in the kernel of $\Gamma$.}, and since~\eqref{eq:hatL decomp} is in essence an energy balance, one should not expect such a balance if cycles are not taken in account. The starting point of MFT is therefore to consider the flux large-deviations~\eqref{eq:L} and decompose it as follows, see for example~\cite{Bertini2015MFT,Renger2018a,MaesNetocny2008,PattersonRengerSharma2021TR}:
\begin{equation}
  \L(\rho,j) = \Psi(\rho,j) + \Psi^*\big(\rho,F(\rho)\big) - F(\rho)\cdot j,
\label{eq:L decomp}
\end{equation}
For all $j\in\RR^{\R_\fw}$ and all $\rho$ for which the force $F(\rho)$ is well defined. The dissipation potentials and force are uniquely given by~\cite[Th.~3.2]{RengerZimmer2021} (see also Remark~\ref{rem:hatPsi}):
\begin{align*}
  F_r(\rho) &\coloneqq  \frac12\log\frac{k_r(\rho)}{k_{\bw(r)}(\rho)}, \\
  \Psi^*(\rho,\zeta) &\coloneqq  2\sum_{r\in\R_\fw}\sqrt{k_r(\rho)k_{\bw(r)}}\big(\cosh(\zeta_r)-1\big), \\
  \Psi(\rho,j) &\coloneqq  2\sum_{r\in\R_\fw} \sqrt{k_r(\rho)k_{\bw(r)}} \Big( \cosh^*\big(\mfrac{j_2}{2\sqrt{k_r(\rho)k_{\bw(r)}}}\big) +1 \Big).
\end{align*}
Recall from~\eqref{eq:reaction rates} that the temperature-dependencies are implicit in the reaction rates $k_r,k_{\bw(r)}$. From the form of the force it becomes clear that \eqref{eq:L decomp} is only valid on the interior of $\fS$, where all reaction rates are positive. 

In case the process $\rho\super{V}(t)$ is in (stochastic) detailed balance, then similar to \eqref{eq:hatL decomp}, we would have $F(\rho)=-\frac12\Gamma\tp\grad\V(\rho)$ \cite[Rem.~3.3]{RengerZimmer2021}, where the $\Gamma\tp$ appears since the forces now act on fluxes. We thus decompose the force $F(\rho)=F^\sym(\rho)+F^\asym(\rho)$ into
\begin{align*}
  F^\sym_r(\rho)  &\coloneqq  -\frac12\Gamma\tp\grad\V(\rho), \quad\text{and}\\
  F^\asym_r(\rho) &\coloneqq  F(\rho)-F^\sym(\rho)= \frac12\log\frac{\kappa_r\pi^{\alpha\super{r}}}{\kappa_{\bw(r)}\pi^{\alpha^{\bw(r)}}},
\end{align*}
where $F^\asym\equiv0$ if (IDB) holds. Together with Theorem~\ref{th:LDP flux} and \eqref{eq:L decomp}, this provides the first decomposition~\eqref{eq:L decomp1}.

The two forces $F^\sym,F^\asym$ are symmetric respectively antisymmetric with respect to time-reversal of the microscopic process; we refer to ~\cite{Bertini2015MFT,RengerZimmer2021,PattersonRengerSharma2021TR} for the details. Observe that $F^\asym\equiv0$ precisely when (IDB) holds. Moreover, the antisymmetric force is constant in $\rho$, which was also observed in all applications in \cite{Bertini2015MFT,KaiserJackZimmer2018,PattersonRengerSharma2021TR}.

\subsection{Generalised orthogonality}

The first decomposition~\eqref{eq:L decomp1} still includes one term that involves both forces: $\Psi^*(\rho,F^\sym(\rho)+F^\asym(\rho))$. This term can be further decomposed by the notion of generalised orthogonality, first described in \cite{KaiserJackZimmer2018}, extended to jump processes in \cite{RengerZimmer2021}, and finally described for arbitrary processes in \cite{PattersonRengerSharma2021TR}. The idea is to generalise Hilbert orthogonality:
\begin{align*}
  \frac12\lVert a+b\rVert^2 = \frac12\lVert a\rVert^2 + \langle a,b\rangle + \frac12\lVert b\rVert^2 = \frac12\lVert a\rVert^2 + \frac12\lVert b\rVert^2 && \text{ whenever } a\perp b.
\end{align*}
A similar expansion can be written down for our dual dissipation potential $\Psi^*$, however since it is not quadratic, the cross product $\langle\cdot,\cdot\rangle$ becomes a non-bilear function $\eta_\rho(\cdot,\cdot)$, one of the potentials need to be ``modified'', and there are two ways to do so:
\begin{align}
  &\Psi^*\big(\rho,F^\sym(\rho)+F^\asym(\rho)\Big) \\
  &\qquad= \Psi^*\big(\rho,F^\sym(\rho)\big) + \eta_\rho\big(F^\asym(\rho),F^\sym(\rho)\big) +  \Psi^*_{F^\asym}\big(\rho,F^\sym(\rho)\big)\\
  &\qquad= \Psi^*\big(\rho,F^\asym(\rho)\big) +\eta_\rho\big(F^\sym(\rho),F^\asym(\rho)\big) +  \Psi^*_{F^\sym} \big(\rho,F^\asym(\rho)\big),
\label{eq:gen orth}
\end{align}
It turns out that the `cross product' $\eta_\rho(F^\sym(\rho),F^\asym(\rho))$ is always zero, which can be interpreted as a nonlinear version of $F^\sym(\rho)\perp F^\asym(\rho)$. We abbreviate the two terms on the right as $\Lambda_\asym^\sym(\rho)$ and $\Lambda_\sym^\asym(\rho)$ respectively.

\begin{theorem} Assume (ICB) and that $\sa_r$ is constant in $r\in\R$, so that by Proposition~\ref{prop:QP} the quasipotential $\V$ is given by \eqref{eq:QP2}. Then for all $\rho$ in the interior of $\fS$:
\begin{align}
  \Psi^*\big(\rho,F^\sym(\rho)+F^\asym(\rho)\Big) &= \Psi^*\big(\rho,F^\sym(\rho)\big) + \Lambda_\sym^\asym(\rho) \notag\\
  &= \Psi^*\big(\rho,F^\asym(\rho)\big) + \Lambda_\asym^\sym(\rho),
  \label{eq:Psis decomp}
\end{align}
with
\begin{align*}
  \Lambda_\sym^\asym(\rho)&\coloneqq \frac12\sum_{r\in\R} e^{-\frac{\sa_r-\se\cdot\alpha\super{r}}{k_B\theta}} \big(\mfrac{\rho}{\pi}\big)^{\alpha\super{r}} \Big( \sqrt{\kappa_r\pi^{\alpha\super{r}}} - \sqrt{  \kappa_{\bw(r)} \pi^{\alpha^{\bw(r)}} }\, \Big)^2, \\
  \Lambda_\asym^\sym(\rho)&\coloneqq \frac12\sum_{r\in\R} \kappa_r\pi^{\alpha\super{r}}\bigg(  \sqrt{e^{-\frac{\sa_r-\se\cdot\alpha\super{r}}{k_B\theta}}\big(\mfrac{\rho}{\pi}\big)^{\alpha\super{r}}} - \sqrt{e^{-\frac{\sa_r-\se\cdot\alpha^{\bw(r)} }{k_B\theta}}\big( \mfrac{\rho}{\pi}\big)^{\alpha^{\bw(r)}} } \,\bigg)^2.
\end{align*}
\end{theorem}

\begin{proof} It was shown in \cite[Prop.~4.2]{RengerZimmer2021} that indeed $\eta_\rho(F^\sym(\rho),F^\asym(\rho))\equiv0$. The remaining terms on the right in \eqref{eq:gen orth} have the explicit form~\cite[Sec.~5]{RengerZimmer2021}:
\begin{align}
  \Lambda_\sym^\asym(\rho)&=\frac12\sum_{r\in\R} \Big( \sqrt{k_r(\rho)} - \sqrt{{\overleftarrow{k}\hspace{-0.3em}}_r(\rho)} \Big)^2, \\
  \Lambda_\asym^\sym(\rho)&=\frac12\sum_{r\in\R} \Big( \sqrt{k_r(\rho)} - \sqrt{{\overleftarrow{k}\hspace{-0.3em}}_{\bw(r)}(\rho)} \Big)^2.
\label{eq:Fishers1}
\end{align}
Here ${\overleftarrow{k}\hspace{-0.2em}}_r(\rho)$ and ${\overleftarrow{k}\hspace{-0.2em}}_{\bw(r)}(\rho)$ are the reaction rates corresponding to the adjoint Markov process $t\mapsto\rho\super{V}(T-t)$ starting from its invariant measure $\Pi\super{V}$. These rates are related to the quasipotential $\V$ via the relation~\cite[Eq.~(42)]{Renger2018a}:
\begin{equation*}
  {\overleftarrow{k}\hspace{-0.2em}}_{\bw(r)}(\rho) = k_r(\rho)e^{\gamma\super{r}\cdot\grad\V(\rho)}.
\end{equation*}
Plugging these and the reaction rates \eqref{eq:reaction rates} into \eqref{eq:Fishers1} completes the proof.
\end{proof}

\subsection{Interpretation}

Combining the decompositions \eqref{eq:L decomp} and \eqref{eq:Psis decomp} and integrating over time yields, using $F^\sym(\rho(t))\cdot\dot w(t)=-\frac12\grad\V(\rho(t)\cdot\Gamma w(t)=-\frac12\frac{d}{dt}\V(\rho(t))$ and $F^\asym(\rho(t))\cdot\dot w(t)=F^\asym\cdot\dot w(t)=\frac{d}{dt} (F^\asym\cdot w(t))$:
\begin{align}
\J_{\lbrack0,T\rbrack}(\rho,w)&=\int_0^T\!\L\big(\rho(t),\dot w(t)\big)\,dt \notag\\
  &= \int_0^T\!\Big\lbrack\underbrace{\Psi\big(\rho(t),\dot w(t)\big) + \Psi^*\big(\rho(t),F^\asym(\rho(t))\big) - F^\asym\big(\rho(t)\big)\cdot \dot w(t)}_{=:\L_{F^\asym}(\rho(t),\dot w(t))} \Big\rbrack\,dt \notag\\
    &\hspace{6em}+ \int_0^T\!\Lambda_\asym^\sym\big(\rho(t)\big)\,dt +\frac12\V(\rho(T))-\frac12\V(\rho(0))
    \label{eq:MFT sym}\\
  &= \int_0^T\!\Big\lbrack\underbrace{\Psi\big(\rho(t),\dot w(t)\big) + \Psi^*\big(\rho(t),F^\sym(\rho(t))\big) - F^\sym\big(\rho(t)\big)\cdot \dot w(t)}_{=:\L_{F^\sym}(\rho(t),\dot w(t))} \Big\rbrack\,dt \notag\\
    &\hspace{6em}+ \int_0^T\!\Lambda_\sym^\asym\big(\rho(t)\big)\,dt - F^\asym\cdot w(T).
    \label{eq:MFT asym}
\end{align}
A priori these decompositions only hold as long as $\rho(t)$ remains in the interior of $\fS$, but since all expressions are well-defined for any $\rho\in\fS$, an approximation argument yields the result for all paths $\rho(t)$ in $\fS$, see for example~\cite[Prop.~5.3 \& 5.4]{RengerZimmer2021}. The two decompositions isolate dissipative and non-dissipative effects in anisothermal reaction networks, and as far as we are aware, this is the first result in this direction.

By convex duality the two expressions $\L_{F^\sym}(\rho(t),\dot w(t))$ and $\L_{F^\asym}(\rho(t),\dot w(t))$ are again non-negative, and their integrals can be interpreted as a large-deviation rate of a modified microscopic process. The modified process corresponding to $\L_{F^\sym}$ is in detailed balance, and hence it is decomposed in a similar fashion as the Onsager-Machlup result~\eqref{eq:hatL decomp} (but now on the level of fluxes). The interpretation of $\L_{F^\asym}$ is a bit more involved. It is shown in \cite{PattersonRengerSharma2021TR} that for simpler examples the path $\L_{F^\asym}(\rho(t),w(t))=0$ is a Hamiltonian system - although for (isothermal) reaction networks this is still an open question.

What can one deduce from~\eqref{eq:MFT sym} and \eqref{eq:MFT asym}? First of all, $\L(\rho(t),\dot w(t))\equiv0$ along the expected path~\eqref{eq:ODE}, so one obtains estimates on the free energy loss and work:
\begin{align*}
  \frac12\V(\rho(T))-\frac12\V(\rho(0)) &= - \int_0^T\!\L_{F^\asym}\big(\rho(t),\dot w(t)\big)\,dt - \int_0^T\!\Lambda_\asym^\sym\big(\rho(t)\big)\,dt\leq0,\\
  F^\asym\cdot w(T)                     &= - \int_0^T\!\L_{F^\sym} \big(\rho(t),\dot w(t)\big)\,dt - \int_0^T\!\Lambda_\sym^\asym\big(\rho(t)\big)\,dt\leq0.
\end{align*}
Naturally, the fact that the quasipotential $\V$ is Lyapunov along the expected macroscopic path is well-known, see for example~\cite{Freidlin2012}.

Second, away from the expected one obtains the estimates:
\begin{align*}
  \frac12\V(\rho(T)) + \int_0^T\!\Lambda_\asym^\sym\big(\rho(t)\big)\,dt &\leq \frac12\V(\rho(0)) + \int_0^T\!\L\big(\rho(t),\dot w(t)\big)\,dt,\\
  F^\asym\cdot w(T) + \int_0^T\!\Lambda_\sym^\asym\big(\rho(t)\big)\,dt  &\leq  \int_0^T\!\L\big(\rho(t),\dot w(t)\big)\,dt.
\end{align*}
Hence a good control of the rate functional and initial condition implies good control of the Fisher information and $\frac12\V(\rho(T)), F^\asym\cdot w(T)$. These estimates are very useful to obtain compactness for paths of bounded rate functional, see for example~\cite{HilderPeletierSharmaTse2020,PeletierRenger2021}.


\section*{Acknowledgements}

This research has been funded by Deutsche Forschungsgemeinschaft (DFG) through grant CRC 1114 ``Scaling Cascades in Complex Systems'', Project Number 235221301, Project C08. The author thanks Luigi delle Site and Upanshu Sharma for their valuable comments.

\bibliographystyle{alphainitials}
\bibliography{library}

\end{document}